\documentclass{article}
\usepackage[a4paper,margin=1.5in,footskip=0.25in]{geometry}
\usepackage{amssymb}
\usepackage{amsmath}
\usepackage{amsthm}
\usepackage{graphicx}
\usepackage{proof}
\usepackage{enumerate}
\usepackage{appendix}
\usepackage{enumitem}
\usepackage{multicol}
\usepackage{subcaption}
\usepackage{hyperref}
\usepackage{xcolor}
\hypersetup{
	colorlinks=true,
	linkcolor=blue,
	citecolor=blue
}

\theoremstyle{plain}

\newtheorem{theorem}{Theorem}

\theoremstyle{definition}

\newtheorem{definition}{Definition}

\theoremstyle{remark}
\newtheorem{example}{Example}

\newcommand{\bs}{\backslash}

\newcommand{\Logic}{\mathrm{MALC}^\ast_\omega}

\title{Extending Action Logic with Omega Iteration}

\author{Tikhon Pshenitsyn}

\begin{document}

\maketitle

\section{Introduction}

Omega algebras introduced in \cite{Cohen00} extend Kleene algebras---abstract counterpart of regular expressions---by omega iteration and thus enable one to reason about both finite and infinite looping processes in system behaviour. While finite Kleene iteration is defined by induction as the least fixed point of the operator $\lambda x. (1 \vee a\cdot x)$, omega iteration is defined in \cite{Cohen00} by the unfolding axiom $x=x\cdot x^\omega$ and the coinduction principle: $x \le y \cdot x+z \Longrightarrow x \le y^\omega+y^\ast \cdot z$. Action algebras is another extension of Kleene algebras \cite{Pratt91}, now with residuals, a.k.a.~implications of the Lambek calculus. They admit relational semantics where a set of states $W$ is fixed and formulae are interpreted as binary relations, i.e.~transitions between states. 

Axiomatic systems concerning omega algebras and omega-regular languages have been actively studied in the recent literature. For example, \cite{DasD24} presents a cyclic proof system where omega iteration is naturally axiomatised by the greatest fixed point of a suitable operator.

We are interested in providing a well-founded proof system for omega-regular expressions that would befriend action algebras with omega iteration. There can be several ways to do so. In this preliminary technical report, we propose one solution where $A^\omega$ is interpreted as infinitary multiplicative conjunction $A \cdot A \cdot A \cdot \ldots$ Interestingly, a similar operation arose in a quite different context in \cite{NicolaiPT24} as vacuous multiplicative quantifier. In Section \ref{section:calculus} we introduce the infinitary action logic with omega iteration $\Logic$ and prove cut admissibility. In Section \ref{section:completeness} we show that $\Logic$ axiomatizes valid inclusions of omega-regular languages. In Section \ref{section:complexity}, we prove the lower $\Pi^1_2$ and the upper $\Pi^2_1$ complexity bounds for provability in $\Logic$.

\section{Multiplicative-Additive Lambek Calculus with Kleene and Omega Iterations}\label{section:calculus}

Following \cite{Cohen00}, we introduce the calculus $\Logic$ using \emph{standard terms}. Formulae have two sorts, corresponding to languages consisting of finite words and to those consisting of infinite ones.

\begin{definition}
	Fix two sets of variables, $\mathrm{Var}_\ast$ and $\mathrm{Var}_\omega$. The syntax of finite formulae $F \in \mathrm{Fm}_{\ast}$ and $\omega$-formulae $\Omega \in \mathrm{Fm}_{\omega}$ is defined by the following grammar in Backus-Naur form:
	\begin{align*}
		F & := \mathrm{Var}_\ast \mid F \cdot F \mid F \bs F \mid F / F \mid F \vee F \mid F \wedge F \mid \Omega / \Omega \mid F^\ast \\
		\Omega & := \mathrm{Var}_\omega \mid F \cdot \Omega \mid F \bs \Omega \mid \Omega \vee \Omega \mid \Omega \wedge \Omega \mid F^\omega
	\end{align*}
	Let $\mathrm{Fm} = \mathrm{Fm}_\ast \cup \mathrm{Fm}_\omega$.
\end{definition}
\begin{definition}
	A sequent is a structure of one of the three forms:
	\begin{enumerate}
		\item $A_1,\ldots,A_n \Rightarrow B$, \quad $A_i,B \in \mathrm{Fm}_\ast$ ($A_1,\ldots,A_n$ is called a type-1 sequence);
		\item $A_1,A_2,\ldots \Rightarrow B$, \quad $A_i \in \mathrm{Fm}_\ast$, $B \in \mathrm{Fm}_\omega$ ($A_1,A_2,\ldots$ is called a type-2 sequence);
		\item $A_1,\ldots,A_n,C \Rightarrow B$, \quad $A_i \in \mathrm{Fm}_\ast$, $B,C \in \mathrm{Fm}_\omega$ ($A_1,\ldots,A_n,C$ is called a type-3 sequence).
	\end{enumerate}
	A \emph{correct sequence} is a sequence of type-$i$ for $i \in \{1,2,3\}$.
\end{definition}
Since sequents---the main syntactic objects we are dealing with---are infinite objects in general, one could expect that the calculus we are about to define potentially has a very high complexity level. 

In the following, $\Pi,\Pi_i,\Gamma_i,\Delta_i$ are type-1 sequences; $\Theta$ can be a sequence of any of the three types. The notation $\langle \Gamma_i \rangle_{i<\alpha}$ is a shorthand for $\Gamma_0,\Gamma_1,\ldots$

\textbf{Axioms}
$$
A \vdash A
$$

\textbf{Left rules}

$$
\infer[(\cdot L)]
{
	\Gamma, A \cdot B, \Theta \vdash E
}
{
	\Gamma, A , B, \Theta \vdash E
}
\qquad
\infer[(\bs L)]
{
	\Gamma, \Pi, B \bs A, \Theta \vdash E
}
{
	\Gamma, A, \Theta \vdash E & \Pi \vdash B
}
\qquad
\infer[(/ L)]
{
	\Gamma, A/D, \Xi, \Theta \vdash E
}
{
	\Gamma, A, \Theta \vdash E & \Xi \vdash D
}
$$
$$
\infer[(\wedge L)]{\Gamma, A_1 \wedge A_2, \Theta \vdash E}{\Gamma, A_i, \Theta \vdash E}
\qquad
\infer[(\vee  L)]{\Gamma, A_1 \vee A_2, \Theta \vdash E}{\Gamma, A_1, \Theta \vdash E & \Gamma, A_2, \Theta \vdash E}
$$
$$
\infer[(^\ast L)]{\Gamma , A^\ast , \Theta \vdash E}{(\Gamma , A^n , \Theta \vdash E)_{n \in\omega}}
\qquad
\infer[(^\omega L)]{\Gamma , A^\omega \vdash E}{\Gamma , \langle A \rangle_{n \in \omega} \vdash E}
$$

We also need a left rule that allows one to apply infinitely many left rules in an infinite sequent.

\begin{definition}\label{definition:Vdash}
	Let $\Theta$ be a correct sequence and let $M$ be a set of correct sequences of the same type. We write $\Theta \Vdash M$ if there is a derivation of $\Theta \vdash x$ from axioms and hypotheses from $\{\Xi \vdash x \mid \Xi \in M\}$ that uses all these hypotheses and only the left rules introduced above ($x$ is a fresh variable).
\end{definition}

\begin{example}
	$A_1 \vee A_2 \Vdash \{A_1,A_2\}$; $A^\ast \Vdash \{A^n \mid n \in \omega\}$; $\Theta \Vdash \Theta$.
\end{example}
Now, the infinitary left rule of $\Logic$ is introduced as follows.
$$
\infer[(L_\omega)]
{
	\langle \Delta_i \rangle_{i<\omega} \vdash C
}
{
	(\Delta_i \Vdash M_i)_{i<\omega}
	&
	\forall (\Gamma_i \in M_i) \: 
	\langle \Gamma_i \rangle_{i<\omega} \vdash C
}
$$

\textbf{Right rules}
$$
\infer[(^\omega R)]
{
	\langle \Gamma_i \rangle_{i<\omega} \vdash A^\omega
}
{
	(\Gamma_i \vdash A)_{i<\omega}
}
$$
$$
\infer[(/ R)]{\Pi \vdash A / B}{\Pi, B \vdash A}
\qquad
\infer[(\bs R)]{\Theta \vdash B \bs A}{B, \Theta \vdash A}
\qquad
\infer[(\cdot R)]{\Gamma, \Theta \vdash A \cdot B}{\Gamma \vdash A & \Theta \vdash B}
$$
$$
\infer[(\wedge R)]{\Theta \vdash A_1 \wedge A_2}{\Theta \vdash A_1 & \Theta \vdash A_2}
\qquad
\infer[(\vee R_i)]{\Theta \vdash A_1 \vee A_2}{\Theta \vdash A_i}
\qquad
\infer[(^\ast R_n)]{\langle \Pi_i \rangle_{i<n} \vdash A^\ast}{(\Pi_i \vdash A)_{i<n}}
$$

Proofs in $\Logic$ are well-founded trees. It is not hard to see that each rule application has either finitely many premises or countably many premises, except for the $(L_\omega)$ rule, which may have at most continuum many premises. 


The idea of $(L_\omega)$ is that, when one deals with an infinite-size sequent, it is possible that it is obtained by infinitely many rules applied simultaneously; however, they affect non-overlapping parts of the sequent and hence do not interfere. 

The rules for omega iteration show that $A^\omega$ is simply an infinitary generalization of the Lambek calculus' multiplicative conjunction. Also, $A^\omega$ is similar to vacuous multiplicative quantifier $\forall x A$ from \cite{NicolaiPT24}. The only difference is that the logic from \cite{NicolaiPT24} admits weakening and exchange.

\begin{example}
	Below we show that the sequent $(q \cdot p/q)^\omega \Rightarrow q \cdot p^\omega$ is provable in $\Logic$.
	\[
	\infer[(^\omega L)]
	{
		(q \cdot p/q)^\omega \Rightarrow q \cdot p^\omega
	}
	{
		\infer[(L_\omega)]
		{
			q \cdot p/q, q \cdot p/q, \ldots \Rightarrow q \cdot p^\omega
		}
		{
			\infer[(\cdot L)]
			{ (q \cdot p/q \vdash {\color{red}x})_{i < \omega}
			}
			{ q , p/q \vdash {\color{red}x}
			}
			&
			\infer[(\cdot R)]
			{
				q , p/q, q , p/q, \ldots \Rightarrow q \cdot p^\omega
			}
			{
				q \Rightarrow q
				&
				\infer[(L_\omega)]
				{
					p/q, q , p/q, q, \ldots \Rightarrow p^\omega
				}
				{
					\infer[]{(p/q,q \vdash {\color{red}x})_{i < \omega}}{p \vdash {\color{red}x} & q \vdash q}
					&
					\infer[({}^{\omega} R)]{p, p, \ldots \Rightarrow p^\omega}{(p \Rightarrow p)_{i<\omega}}
				}
			}
		}
	}
	\]
	Hereinafter we exploit the following notation: instead of writing $\Delta \Vdash M$, we place a derivation of $\Delta \vdash x$ from $\{\Xi \vdash x \mid \Xi \in M\}$ as in Definition \ref{definition:Vdash}. The fresh variable $x$ is coloured red above. This make the calculus look analytic.
\end{example}
\begin{example}
	Below we show that the sequent $p^n \Rightarrow p^\omega/p^\omega$ is provable in $\Logic$ for any $n \in \omega$.
	\[
	\infer[(/R)]
	{
		p^n \Rightarrow p^\omega/p^\omega
	}
	{
		\infer[(^\omega L)]
		{
			p^n,p^\omega \Rightarrow p^\omega
		}
		{
			\infer[({}^\omega R)]
			{
				p,p,p,\ldots \Rightarrow p^\omega
			}
			{
				(p \Rightarrow p)_{i < \omega}
			}
		}
	}
	\]
\end{example}

\begin{definition}
	Distinguished formula occurrences in antecedents/succedents of conclusions of the left/right rules are called \emph{principal}. In an application of $(L_\omega)$, we assume that $\Delta_i \Vdash M_i$ is replaced by a derivation of $\Delta_i \vdash x$ from $\{\Xi \vdash x \mid \Xi \in M\}$, as in Definition \ref{definition:Vdash}. If this derivation is not an axiom, then a distinguished formula in its last rule application (contained in $\Delta_i$) is also called \emph{principal} in $(L_\omega)$'s rule application. (Thus an application of $(L_\omega)$ can have up to countably many principal formulae.)
\end{definition}

The cut rule is admissible in $\Logic$. To prove this, we need to define a generalized ``mix'' rule that allows one to cut infinitely many formulae at once. (A similar argument occurs in \cite{NicolaiPT24}.) The rule looks as follows for $\alpha \le \omega$:
$$
\infer[(\mathrm{mix}_\alpha)]
{
	\Sigma, \langle \Pi_i , \Psi_i \rangle_{i<\alpha} \vdash C
}
{
	(\Pi_i \vdash A_i)_{i<\alpha}
	&
	\Sigma, \langle A_i , \Psi_i \rangle_{i<\alpha} \vdash C
}
$$
It is required that the sizes of $A_i$ are bounded by some $K \in \omega$.

\begin{theorem}
	The $(\mathrm{mix}_\alpha)$ rule is admissible.
\end{theorem}
\begin{proof}[Proof sketch]
	The proof goes by induction on (a) the supremum of the depths of $A_i$'s, (b) on $\alpha$, (c) on the smallest quasiheight of a derivation ending with $\Sigma, \langle A_i , \Psi_i \rangle_{i<\alpha} \vdash C$, (d) on the supremum of the smallest heights of derivations of $(\Pi_i \vdash A_i)_{i<\alpha}$.
	
	If $\alpha$ is finite, one can split $(\mathrm{mix}_\alpha)$ into an application of $(\mathrm{mix}_{\alpha-1})$ (to which we apply the induction hypothesis) and that of $(\mathrm{cut})$. The cut rule can be eliminated in the same manner as in \cite{Palka07}---unless the formula being cut is an omega iteration. Let us consider the case where $A^\omega$ is principal in both premises of cut:
	\[
	\infer[(\mathrm{cut})]
	{
		\Gamma,\langle \Gamma_i \rangle_{i<\omega} \Rightarrow C
	}
	{
		\infer[(^\omega R)]
		{
			\langle \Gamma_i \rangle_{i<\omega} \vdash A^\omega
		}
		{
			(\Gamma_i \vdash A)_{i<\omega}
		}
		&
		\infer[(L)]
		{
			\Gamma, A^\omega \Rightarrow C
		}
		{
			\Gamma \Vdash \Gamma & A^\omega \Vdash A,A,A,\ldots
			&
			\phantom{abc}
			&
			\Gamma,A,A,A,\ldots \Rightarrow C
		}
	}
	\]
	This is transformed into an application of $(\mathrm{mix}_\omega)$:
	\[
	\infer[(\mathrm{mix}_\omega)]
	{
		\Gamma,\langle \Gamma_i \rangle_{i<\omega} \Rightarrow C
	}
	{
		(\Gamma_i \vdash A)_{i<\omega}
		&
		\Gamma,A,A,A,\ldots \Rightarrow C
	}
	\]
	Here the size of $A$'s is less than the size of $A^\omega$, so the induction hypothesis can be applied.
	
	Now, let us discuss admissibility of 
	\[
	\infer[(\mathrm{mix}_\omega)]
	{
		\Sigma, \langle \Pi_i , \Psi_i \rangle_{i<\omega} \vdash C
	}
	{
		(\Pi_i \vdash A_i)_{i<\omega}
		&
		\Sigma, \langle A_i , \Psi_i \rangle_{i<\omega} \vdash C
	}
	\]
	Note that each $A_i$ is finite.	Fix a derivation of $\Sigma, \langle A_i , \Psi_i \rangle_{i<\omega} \vdash C$. If $\Sigma, \langle A_i , \Psi_i \rangle_{i<\omega} \vdash C$ is obtained by a right rule, one can permute it with $(\mathrm{mix}_\omega)$. So let us assume that it is obtained by means of $(L_\omega)$. This rule application splits $\Sigma, \langle A_i , \Psi_i \rangle_{i<\omega}$ into countably many parts $\Xi_0,\Xi_1,\ldots$ The $k$-th part has the form $\Xi_k =  \Gamma^k_{0}, A_{i_k}, \Delta^k_0, \ldots, A_{i_{k+1}-1}, \Delta^k_{i_{k+1}-i_k-1}$ for $i_k \le i_{k+1}$. The application of $(L_\omega)$ looks as follows:
	$$
	\infer[(L_\omega)]
	{
		\langle \Xi_i \rangle_{i<\omega} \vdash C
	}
	{
		(\Xi_i \Vdash M_i)_{i<\omega}
		&
		\forall (\Psi_i \in M_i) \: 
		\langle \Psi_i \rangle_{i<\omega} \vdash C
	}
	$$
	
	Consider $\Xi_k$ for some fixed $k$. For the sake of simplicity, let $k=0,i_0=0, i_1=h$, i.e.~$\Xi_0 = \Gamma_{0}, \langle A_{i}, \Delta_i \rangle_{i<h}$ (we also omit the superscript $0$). Consider the following mix application:
	\[
	\infer[]
	{
		\Gamma_{0}, \langle \Pi_{i}, \Delta_i \rangle_{i<h} \vdash x
	}
	{
		(\Pi_i \vdash A_i)_{i<h}
		&
		\Gamma_{0}, \langle A_{i}, \Delta_i \rangle_{i<h} \vdash x
	}
	\]
	Unfortunately, we cannot apply the induction hypothesis because the rightmost premise is not derivable in $\Logic$. We know, however, that it is derivable from the hypotheses from $\mathcal H = \{\Psi \vdash x \mid \Psi \in M_0\}$. We lift the mix rule application to the sequents from $\mathcal H$. Finally, we introduce a series of new mix rules applied to sequents of the form $\langle \Psi_i \rangle_{i<\omega} \vdash C$ for all $\Psi_i \in M_i$. 
	
\end{proof}

Another, notationally interesting way to introduce the cut rule is as follows:
$$
\infer[(\mathrm{cut})]
{
	\Pi \Vdash A
}
{
	\Pi \vdash A
}
$$

Cut admissibility implies that the rules $(\vee L)$, $(\wedge R)$, $(\cdot L)$, $(\bs R)$, $(/ R)$, $(^\ast L)$, and $(^\omega L)$ are invertible.

\section{Completeness for omega-regular expressions}\label{section:completeness}

An omega-regular expression is a term in the language $\cdot, \vee, ^\ast, ^\omega$. A variable $x$ is interpreted as the language $\{x\}$ and the operations are interpreted in a standard manner.
\begin{theorem}
	Let $\alpha,\beta$ be omega-regular expressions. Then $\alpha \subseteq \beta$ holds iff $\Logic$ proves $\alpha \vdash \beta$.
\end{theorem}
\begin{proof}
	If $\alpha,\beta$ do not contain $\omega$, then this follows from the corresponding result for infinitary action logic. Now, using invertibility of $(\cdot L)$, $(\vee L)$, $(^\ast L)$, and $(^\omega L)$ we can ensure that $\alpha \vdash \beta$ is equiderivable with the sequents of the form $w \vdash \beta$ for all $w \in \alpha$ (infinite words). In the latter sequent, only right rules can be applied, which ensures that $w \in \beta$. This proves the less trivial completeness part.
\end{proof}

If one adds $\bs,/$ along with their standard language semantics, however, the result does not hold anymore. Namely, consider the sequent $q \cdot (q \bs (p \cdot q))^\omega \vdash p^\omega$. It is true in language semantics but not provable.

\section{Complexity}\label{section:complexity}

\begin{theorem}
	Provability in $\Logic$ is in $\Pi^2_1$.
\end{theorem}
\begin{proof}
	A sequent $s$ is provable iff for each set $S$ of sequents, if $S$ contains all axioms and is closed under rule applications, then $S$ contains $s$. Sequents of $\Logic$ are either finite (then they can be encoded by natural numbers) or countably infinite (then they can be encoded by subsets of $\omega$). Thus, the quantifier ``for each set $S$ of sequents'' is third-order. All the remaining properties are at most second-order.
\end{proof}

\begin{theorem}
	Provability in $\Logic$ is $\Pi^1_2$-hard.
\end{theorem}
\begin{proof}
	We reduce a $\Pi^1_2$-complete problem from \cite{Finkel09}. Given context-free grammars $G_1,\ldots,G_{2n}$ over the alphabet $\{a,b\}$, we need to check whether the following holds:
	\vspace{-0.3cm}
	\begin{equation}\label{equation:totality}
		\bigcup\limits_{i=1}^n L(G_{2i-1}) \cdot L(G_{2i})^{\omega}
		= \{a,b\}^{\omega}
	\end{equation}
	
	Given a context-free grammar $G_i$, one can find finitely many formulae $T_{ij}(a),T_{ij}(b)$ of the form $p$, $p/q$ or $(p/q)/r$ and formulae $s_i \in \mathrm{Var}$ such that $a_1\ldots a_n \in L(G_i)$ $\Longleftrightarrow$ $T_{ij_1}(a_1), \ldots, T_{ij_n}(a_n) \vdash s_i$ is provable in $\Logic$. 
	Without loss of generality, let $T_i(c)$ and $T_j(d)$ not contain common variables whenever $i \ne j$ and $c,d \in \{a,b\}$.
	
	We claim that (\ref{equation:totality}) holds iff the following sequent is provable in $\Logic$ ($q$ is a fresh variable).
	\[
	\left(
	\bigwedge\limits_{i=1}^{2n}\bigwedge\limits_j T_{ij}(a)
	\vee
	\bigwedge\limits_{i=1}^{2n}\bigwedge\limits_j T_{ij}(b)
	\right)^\omega
	\vdash
	\bigvee_{i=1}^n \left( s_{2i-1} \cdot s_{2i}^{\omega} \right)
	\]
	By invertibility of $(^\omega L)$ and $(\vee L)$, derivability of the latter sequent is equivalent to that of the following sequent for each $w = w_0w_1 \ldots \in \{a,b\}^\omega$:
	\[
	\left\langle\bigwedge\limits_{i=1}^{2n} \bigwedge\limits_j T_{ij}(w_k)\right\rangle_{k<\omega}
	\vdash 
	\bigvee_{i=1}^n \left( s_{2i-1} \cdot s_{2i}^{\omega} \right)
	\]
	Since none of the formulae is copied in any possible derivation of the latter sequent, one can assume that all the conjunctions in the antecedents are principal in the last rule application and that then the right-hand side disjunction and product are principal. (Note that there is no product in the antecedent.) Thus the above sequent is equiderivable with the following one for some $i_k,j_k,l$:
	\[
	\left\langle T_{i_kj_k}(w_k)\right\rangle_{k<\omega}
	\vdash 
	s_{2l-1} \cdot s_{2l}^{\omega}
	\]
	which, in turn, is equiderivable with the sequents 
	\[
	\left\langle T_{i_kj_k}(w_k)\right\rangle_{k<n_0}
	\vdash 
	s_{2l-1}
	\qquad
	\left\langle T_{i_kj_k}(w_k)\right\rangle_{n_p \le k < n_{p+1}}
	\vdash 
	s_{2l}
	\]
	for some $n_0 \le n_1 \le \ldots$. Derivability of the first sequent is equivalent to the fact that $w_{0} \ldots w_{n_0-1} \in L(G_{2l-1})$, while derivability of the second one implies that $w_{n_p} \ldots w_{n_{p+1}-1} \in L(G_{2l})$ for all $p \ge 0$. Thus, $w \in L(G_{2l-1}) \cdot L(G_{2l})^\omega$. This indeed implies that derivability of the original sequent is equivalent to (\ref{equation:totality}) being the case.

\end{proof}

\subsection*{Funding} 

This work was supported by the Russian Science Foundation under grant no.~23-11-00104, \href{https://rscf.ru/en/project/23-11-00104/}{https://rscf.ru/en/project/23-11-00104/}.

\bibliographystyle{plain}
\bibliography{bibliography}

\end{document}